\documentclass{sigplanconf}

\usepackage[cmex10]{amsmath}
\usepackage{amsfonts,amstext,amssymb,stmaryrd,mathrsfs,calc,amsthm,xspace,mathdots,mathtools}
\usepackage{amsthm}
\usepackage{stfloats}
\usepackage{microtype}
\usepackage{bm}
\usepackage{graphics}
\usepackage{enumerate}
\usepackage{macros}
\usepackage{complexity}
\usepackage{todonotes}
\usepackage{color, colortbl}
\usepackage{url}
\usepackage{algorithm}
\usepackage{algorithmic}

\definecolor{Gray}{gray}{0.7}

\begin{document}

\setlength{\pdfpageheight}{\paperheight}
\setlength{\pdfpagewidth}{\paperwidth}

\conferenceinfo{CSL-LICS 2014}{July 14--18, 2014, Vienna, Austria}
\copyrightyear{2014}
\copyrightdata{978-1-4503-2886-9}
\doi{2603088.2603092}

\title{Subclasses of Presburger Arithmetic and the\\ Weak EXP
  Hierarchy}

\authorinfo{Christoph Haase\thanks{The author is supported by the
    French Agence Nationale de la Recherche (ANR), \textsc{ReacHard}
    (grant ANR-11-BS02-001). }}  
           {Laboratoire Sp\'ecification et
             V\'erification (LSV), CNRS\\ \'Ecole Normale Sup\'erieure (ENS) de
             Cachan, France} 
           {haase@lsv.ens-cachan.fr}

\maketitle

\begin{abstract}
  It is shown that for any fixed $i>0$, the $\Sigma_{i+1}$-fragment of
  Presburger arithmetic, \ie, its restriction to $i+1$ quantifier
  alternations beginning with an existential quantifier, is complete
  for $\phsigma^{\EXP}_{i}$, the $i$-th level of the weak EXP
  hierarchy, an analogue to the polynomial-time hierarchy residing
  between $\NEXP$ and $\EXPSPACE$. This result completes the
  computational complexity landscape for Presburger arithmetic, a line
  of research which dates back to the seminal work by Fischer \& Rabin
  in 1974. Moreover, we apply some of the techniques developed in the
  proof of the lower bound in order to establish bounds on sets of
  naturals definable in the $\Sigma_1$-fragment of Presburger
  arithmetic: given a $\Sigma_1$-formula $\Phi(x)$, it is shown that
  the set of non-negative solutions is an ultimately periodic set
  whose period is at most doubly-exponential and that this bound is
  tight.
\end{abstract}




\section{Introduction}

\emph{Presburger arithmetic} is the first-order theory of the
structure ${\langle \mathbb{N},0,1,+,<\rangle}$. This theory was shown
to be decidable by Presburger in his seminal paper in 1929 by
providing a quantifier-elimination procedure~\cite{Pre29}. Presburger
arithmetic is central to a vast number of different areas in computer
science and is often employed as a tool for showing decidability and
complexity results.

The central decision problem for Presburger arithmetic is
\emph{validity}, \emph{i.e.}, to decide whether a given sentence is
true with respect to the standard interpretation in arithmetic. The
two most prominent ways to decide validity are either
quantifier-elimination based~\cite{Coo72} or automata
based~\cite{WB95,Kla08,DGH12}. Any decision procedure for Presburger
arithmetic is inherently tied to the computational complexity of
Presburger arithmetic; for that reason the complexity of Presburger
arithmetic has extensively been studied in the literature from the
1970's onwards. In order to fully capture the computational complexity
of Presburger arithmetic, Berman even introduced the $\STA$ measure on
the complexity of a decision problem, since Presburger arithmetic
\emph{``may not have precise complexity characterisations in terms of
  the usual time and tape measures''}~\cite{Ber80}. The class
$\STA(s(n),t(n),a(n))$ is the class of all problems of length $n$ that
can be decided by an alternating Turing machine in space $s(n)$ and
time $t(n)$ using $a(n)$ alternations, where ``$\ast$'' acts as a
wildcard in order to indicate an unbounded availability of a certain
resource. Based on the work by Fischer \& Rabin~\cite{FR74} and
Ferrante \& Rackoff~\cite{FR75}, Berman established the following
result.

\begin{proposition}[Berman~\cite{Ber80}]
  \label{thm:berman}
  Presburger arithmetic is complete for $\STA(\ast,2^{2^{n^{\bigo(1)}}},n)$.
\end{proposition}

In terms of the usual time and space measures, this settles Presburger
arithmetic between $\textsf{2}$-$\NEXP$ and
$\textsf{2}$-$\EXPSPACE$. Despite these high computational costs, on
the positive side when looking at fragments Presburger arithmetic
becomes more manageable. There are two dimensions in which we can
constraint formulas in order to obtain fragments of Presburger
arithmetic: the number of quantifier alternations and the number of
variables in each quantifier block. For $i,j\in \N\cup \{\ast\}$, let
PA($i$,$j$) denote the set of formulas of the $\Sigma_i$-fragment of
Presburger arithmetic\footnote{All results obtained are symmetric when
  considering $\Pi_i$-formulas.} such that at most $j$ different
variables occur in each quantifier block, where ``$\ast$'' is used as
a wildcard for an unbounded number. Hence,
Proposition~\ref{thm:berman} characterises the computational
complexity of PA($\ast$,$\ast$) with $n$ being the number of symbols
required to write down the formula. Subsequently, PA and PA($i$)
abbreviate PA($\ast$,$\ast$) and PA($i$,$\ast$), respectively.

One of the most prominent fragments of Presburger arithmetic is its
existential or quantifier-free fragment, which is computationally not
more expensive than standard Boolean satisfiability.
\begin{proposition}[Scarpellini~\cite{Sca84}, Borosh \& Treybing~\cite{BT76}]
  For any fixed $j\in \N$, {PA($1$,$j$)} is in
  $\P$~\cite{Sca84}. {PA($1$)} is $\NP$-complete~\cite{BT76}.
\end{proposition}

Due to its comparably low computational complexity, quantifier-free
Presburger arithmetic is the fragment that is most commonly found in
application areas which aim at a practical impact. The existential
fragment of Presburger arithmetic can even be extended with a full
divisibility predicate while retaining
decidability~\cite{Lip78,Lip81}.

Another subclass of Presburger arithmetic which
has extensively been studied is obtained by allowing for an arbitrary
but fixed number of quantifier alternations.
\begin{proposition}[Gr\"adel~\cite{Gra88}, Sch\"oning~\cite{Scho97}, Reddy \& Loveland~\cite{RL78}]
  For any fixed $i>0$ and $j>2$, PA($i+1$,$j$) is
  $\phsigma_{i}^\P$-complete\footnote{In order to establish hardness,
    $j>2$ is only required on the innermost
    quantifier.}~\cite{Gra88,Scho97}. PA($i$) is in
  $\STA(\ast,2^{n^{\bigo(i)}},i)$~\cite{RL78}.
\end{proposition}
Thus, when fixing the number of quantifier alternations, the
complexity of Presburger arithmetic decreases roughly by one exponent,
and when additionally fixing the number of variables, we obtain every
level of the polynomial-time hierarchy. Notice that there is an
obvious gap: a completeness result for Presburger arithmetic with a
fixed number of quantifier alternations and an arbitrary number of
variables in each quantifier block is missing.

The study of lower bounds for PA($i$) goes back to the work of
F\"urer~\cite{F82}, who showed a $\NEXP$ lower bound for some fixed
$i>1$. Later, Gr\"adel showed $\NEXP$-hardness and $\EXPSPACE$
membership of PA(2), but tight lower and upper bounds for the whole
class of PA($i$) formulas have not yet been established. The purpose
of the first part of this paper is to close this gap and establish the
following theorem.
\begin{theorem}
  \label{thm:main}
  For any fixed $i>0$, the $\Sigma_{i+1}$-fragment of Presburger
  arithmetic is $\phsigma^{\EXP}_{i}$-complete.
\end{theorem}
Here, $\phsigma^{\EXP}_i$ denotes the $i$-th level of the \emph{weak
  EXP hierarchy}~\cite{Hem89}, an analogue to the polynomial-time
hierarchy~\cite{Sto76} residing between $\NEXP$ and $\EXPSPACE$; a
formal definition will be provided in
Section~\ref{sec:weth}. Equivalently, we obtain that PA($i+1$) is
complete for $\STA(\ast,2^{n^{\bigo(i)}},i)$. Determining the precise
complexity of the $\Sigma_i$-fragment of Presburger arithmetic for a
fixed $i$ has been listed as a problem that \emph{``deserves to be
  investigated''} by Compton \&
Henson~\cite[Prob.\ 10.14]{CH90}. However, as pointed out
in~\cite{CH90}, their generic methods for proving lower bounds do not
seem to be applicable to this fragment, and our hardness result is
based on rather specific properties of, for instance, distributions of
prime numbers.

The second part of the paper diverts from the first part and focuses
on the $\Sigma_1$-fragment of Presburger arithmetic. More
specifically, we consider sets of naturals definable by formulas in
the $\Sigma_1$-fragment of Presburger arithmetic open in one
variable. Given a $\Sigma_1$-formula $\Phi(x)$, denote by
$\eval{\Phi(x)}$ the set of those $a\in \mathbb{N}$ such that
replacing $x$ with $a$ in $\Phi(x)$ is valid. It is well-known that
$\eval{\Phi(x)}$ is an ultimately periodic set, see
\emph{e.g.}~\cite{Bes02}. A set $N\subseteq \mathbb{N}$ is
\emph{ultimately periodic} if there exists a \emph{threshold}
$t\in\mathbb{N}$, a \emph{base} $B\subseteq \{0,\ldots t-1\}$, a
\emph{period $p\in\mathbb{N}$}, and a set of \emph{residue classes}
$R\subseteq \{0,\ldots p-1\}$ such that $N=U(t,p,B,R)$ with
\begin{align*}
  U(t,p,B,R) & \defeq B \cup \left\{ t + r + kp : r\in R, k\ge 0 \right\}.
\end{align*}

Given a $\Sigma_1$-formula $\Phi(x)$, by applying some insights from
the first part, we can establish a doubly-exponential upper bound on
the period of the ultimately periodic set equivalent to
$\eval{\Phi(x)}$ and show that this bound is tight, which is captured
by the second main theorem of this paper.
\begin{theorem}
  \label{thm:qfpa-projections}
  There exists a family of $\Sigma_1$-formulas of Presburger
  arithmetic $(\Phi_n(x))_{n>0}$ such that each $\Phi_n(x)$ is a
  {PA($1$,$O(n)$)} formula with $\abs{\Phi_n(x)}\in O(n^2)$ and
  $\eval{\Phi_n(x)}$ is an ultimately periodic set with period $p_n\in
  2^{2^{\Omega(n)}}$. Moreover for any $\Sigma_1$-formula $\Phi(x)$,
  we have $\eval{\Phi(x)}=U(t,p,B,R)$ such that $t\in
  2^{\poly(\abs{\Phi(x)})}$ and $p\in 2^{2^{\poly(\abs{\Phi(x)})}}$.
\end{theorem}

The most interesting part about this theorem is the doubly-exponential
lower bound of the period of ultimately periodic sets definable by
PA($1$) formulas. Establishing bounds on constants of ultimately
periodic sets naturally occurs when analysing the computational
complexity of decision problems for infinite-state
systems~\cite{GHOW12} or in formal language theory~\cite{Huy84}. For
instance, analysing such bounds has been crucial in order to obtain
optimal complexity results for model-checking problems of a class of
one-counter automata in~\cite{GHOW12}. In more detail,
in~\cite{GHOW12} it has been shown that the set of non-negative
weights of paths between two nodes in a weighted graph is ultimately
periodic with a period that is at most \emph{singly-exponential}
bounded. A result by Seidl \emph{et al.}~\cite{SSMH04} on Parikh
images of non-deterministic finite-state automata implicitly states
that those ultimately periodic sets are definable in the
$\Sigma_1$-fragment of Presburger arithmetic. It would thus be
desirable to establish a generic upper bound for ultimately periodic
sets definable in PA($1$) yielding the same optimal bounds. In this
context, Theorem~\ref{thm:qfpa-projections} provides a negative result
in that it shows that a general bound on ultimately periodic sets
definable in PA($1$) cannot yield the optimal bounds required for
natural concrete ultimately periodic sets like those considered
in~\cite{GHOW12}.

This paper is structured as follows. In
Section~\ref{sec:preliminaries} we provide most of the formal
definitions required in this paper; however the reader is expected to
have some level of familiarity with standard notions and concepts from
linear algebra, integer programming, first-order logic and
computational complexity. Even though we provide a slightly more
elaborated account on succinct encodings via Boolean circuits, it will
be beneficial to the reader to be familiar with Chapters~8 and 20 in
Papadimitriou's book on computational
complexity~\cite{Pap94}. Section~\ref{sec:complexity} is then going to
establish the lower and upper bounds of Theorem~\ref{thm:main}, and
Theorem~\ref{thm:qfpa-projections} is shown in
Section~\ref{sec:derived-results}. The paper concludes in
Section~\ref{sec:conclusion}. Subsequent to the bibliography, a proof
of a technical characterisation of the weak EXP hierarchy is outlined
in the appendix for the sake of completeness.

\section{Preliminaries}
\label{sec:preliminaries}

\subsection{General notation} By $\mathbb{Z}$ and $\N$ we denote the set
of integers and natural numbers, respectively. We will usually use
$a,b,c$ for numbers in $\Z$ and $\N$. Given $a\in \N$, we define
$[a]\defeq \{0,\ldots a-1\}$. Given sets $M,N\subseteq \N$, as is
standard $M+N\defeq \{ m + n : m\in M, n\in N\}$ and $M\cdot N\defeq
\{ mn : m\in M, n\in N\}$. Moreover, we will use standard notation for
integer intervals and, \emph{e.g.}, for $a\le b\in \N$ denote by
$[a,b)$ the set $\{a,\ldots b-1 \}$. For vectors
  $\vec{a}=(a_1,\ldots,a_n)\in \Z^n$, we will denote by
  $\norm{\vec{a}}$ the \emph{norm of $\vec{a}$}, which is the maximum
  absolute value of all components of $\vec{a}$, \emph{i.e.},
  $\norm{\vec{a}}\defeq \max\{\abs{a_i}\}_{1\le i\le n}$.  For
  $m\times n$ integer matrices $A$, $\norm{A}$ denotes the maximum
  absolute value of all components of $A$. Finally, given a set
  $M\subseteq \Z^n$, we denote by $\norm{M}$ the maximum of the norm
  of all elements of $M$. All functions in this paper are assumed to
  map non-negative integers to non-negative integers.  Unless stated
  otherwise, we assume all integers in this paper to be encoded in
  binary, \emph{i.e.}, the size or length to write down $a\in \Z$ is
  $\bigo(\log \abs{a})$.

\subsection{Presburger Arithmetic} Usually, $x,y,z$ will denote
first-order variables, and $\vec{x},\vec{y},\vec{z}$ vectors or tuples
of first-order variables. Let $\vec{x}=(x_1,\ldots,x_n)$ be an
$n$-tuple of first-order variables. In this paper, formulas of
Presburger arithmetic are standard first-order formulas over the
structure ${\langle \mathbb{N},0,1,+,<\rangle}$ obtained from atomic
expressions of the form $p(\vec{x})< b$, where $p(\vec{x})$ is a
linear multivariate polynomial with integer coefficients and absolute
term zero, and $b\in \Z$. If the dimension of $\vec{x}$ is clear from
the context, for brevity we will often omit stating it explicitly.
Let $\vec{a}=(a_1,\ldots,a_n)\in \N^n$ and $\Phi(\vec{x})$ be open in
the first-order variables $\vec{x}$, we denote by
$\Phi(\vec{a}/\vec{x})$ the closed formula obtained from replacing
every occurrence of $x_i$ in $\Phi(\vec{x})$ with $a_i$. By
$\eval{\Phi(\vec{x})}$ we denote the set $\{ \vec{a}\in \N^n:
\Phi(\vec{a}/\vec{x})$ is valid$\}$. The \emph{size} $\abs{\Phi}$ of a
formula of Presburger arithmetic is defined as the number of symbols
required to write it down, and the \emph{norm} $\norm{\Phi}$ is the
largest absolute value of all constants occurring in $\Phi$. 

\begin{remark}
  For notational convenience, when stating concrete formulas we will
  permit ourselves to use atomic formulas $p(\vec{x})< q(\vec{x})$ for
  linear polynomials $p(\vec{x}), q(\vec{x})$. Moreover, all results
  on the complexity of validity of formulas of Presburger arithmetic
  carry over if we assume unary encoding of numbers, since binary
  encoding of numbers can be ``simulated'' by the introduction of
  additional first-order variables and repeated multiplication by two,
  causing only a sub-quadratic blowup in the formula size. In
  addition, an equality predicate ``$=$'' can be expressed in terms of
  $<$ causing a linear blowup, since $x=y \leftrightarrow x<y+1 \wedge
  y<x+1$. Likewise, $x>y$ and $x<y<z$ abbreviate $y<x$ and $x<y \wedge
  y<z$, respectively.
\end{remark}

\subsection{Semi-Linear Sets and Systems of Linear Diophantine Inequalities}
A central result due to Ginsburg and Spanier states that the sets of
natural numbers definable by a formula of Presburger arithmetic open
in $n$ variables are the \emph{$n$-dimensional semi-linear
  sets}~\cite{Gin66}, which we just call semi-linear sets if the
dimension is clear from the context. A semi-linear set is a finite
union of \emph{linear sets}. The latter are defined in terms of a
\emph{base vector} $\vec{b}\in \mathbb{N}^n$ and a finite set of
\emph{period vectors} $P=\{\vec{p}_1,\ldots \vec{p}_k\}\subseteq
\mathbb{N}^n$, and define the set
\begin{align*}
  L(\vec{b};P)\defeq \vec{b} + \lambda_1\vec{p}_1 + \cdots \lambda_k \vec{p}_k, 
  ~\lambda_i\in \N, 1\le i\le k.
\end{align*}

Let $A$ be an $m\times n$ integer matrix and $\vec{c}\in \Z^m$. A
\emph{system of linear Diophantine inequalities} is given as $S :
A\vec{x} \ge \vec{c}$. The \emph{size $\abs{S}$ of $S$} is the number
of symbols required to write down $S$ assuming binary encoding of
numbers. The \emph{set of positive solutions of $S$} is denoted by
$\eval{S}\subseteq \N^n$ and is the set of all $n$-tuples such that
the inequalities in every row of $S$ hold.

The following proposition is due to Frank \& Tardos and establishes a
strongly polynomial-time algorithm for the feasibility problem of a
system of linear Diophantine inequalities in a fixed dimension, \ie,
deciding whether $\eval{S}\neq \emptyset$.
\begin{proposition}[Frank \& Tardos~\cite{FT87}]
  \label{prop:frank-tardos-ip}
  Let $S:A\vec{x}\ge \vec{c}$ be a system of linear Diophantine
  inequalities such that $A$ is an $m\times n$ matrix. Then
  feasibility of $S$ can be decided using $n^{2.5n+o(n)}\abs{S}$
  arithmetic operations and space polynomial in $\abs{S}$.
\end{proposition}

When we are interested in representing the set of all solutions of
$S$, we will employ the following proposition, which provides bounds
on the semi-linear representation of $\eval{S}$ and is a consequence
of Corollary~1 in~\cite{Pot91}.

\begin{proposition}[Pottier~\cite{Pot91}]
  \label{lem:pottier-bound}
  Let $S:A\vec{x}\ge \vec{c}$ be a system of linear Diophantine
  inequalities such that $A$ is an $m\times n$ matrix. Then
  $\eval{S}=\bigcup_{i\in I}L(\vec{b}_i;P_i)$ such that for all $i\in
  I$,
  \begin{align*}
    \norm{\vec{b}_i}, \norm{P_i}\le (n\norm{A}+\norm{\vec{c}}+2)^{m+n}.
  \end{align*}
\end{proposition}

\subsection{Time Hierarchies} 
\label{sec:weth} Let us recall the
definitions of the \emph{polynomial-time hierarchy} $\PH$~\cite{Sto76}
and the \emph{weak EXP hierarchy} \textsf{EXPH}~\cite{Hem89} in terms
of oracle complexity classes. As usual,
\begin{align*}
  \P& =\bigcup_{k>0}\DTIME(n^k) & \EXP &=\bigcup_{k>0}\DTIME(2^{n^k})\\
  \NP& =\bigcup_{k>0}\NTIME(n^k) & \NEXP &=\bigcup_{k>0}\NTIME(2^{n^k}).
\end{align*}
The aforementioned time hierarchies are now defined as
\begin{align*}
  \phsigma_0^\P \defeq \phpi_0^\P & \defeq \P & \phsigma_0^{\EXP}\defeq \phpi_0^{\EXP}&\defeq \EXP\\
  \phsigma_{i+1}^\P &\defeq \NP^{\phsigma_i^\P} & \phsigma_{i+1}^{\EXP}&\defeq \NEXP^{\phsigma_i^\P}\\
  \phpi_{i+1}^\P &\defeq \coNP^{\phsigma_i^\P} & \phpi_{i+1}^{\EXP}&\defeq \coNEXP^{\phsigma_i^\P}\\
  \PH\ & \defeq \bigcup_{i\ge 0}\phsigma_i^\P & \textsf{EXPH} & \defeq \bigcup_{i\ge 0} \phsigma_i^{\EXP}.
\end{align*}
For our lower bounds, we will rely on the following equivalent
characterisation of $\phsigma^{\EXP}_{i}$.
\begin{lemma}
  \label{lem:wexp-hierarchy-characterisation}
  For any $i>0$, a language $L\subseteq \{0,1\}^*$ is in
  $\phsigma^{\EXP}_i$ iff there exists a polynomial $q$ and a predicate
  $R\subseteq (\{0,1\}^*)^{i+1}$ such that for any $w\in \{0,1\}^n$,
  \begin{multline*}
    w\in L \text{ iff } \exists w_1\in \{0,1\}^{2^{q(n)}}.\forall w_2\in
    \{0,1\}^{2^{q(n)}}\cdots\\\cdots Q_i w_i\in \{0,1\}^{2^{q(n)}}.R(w,w_1,\ldots, w_i)
  \end{multline*}
  and $R(w,w_1,\ldots,w_i)$ can be decided in deterministic polynomial
  time.
\end{lemma}
Despite being in the spirit of an elementary result on computational
complexity, the author was unable to find a formal proof of
Lemma~\ref{lem:wexp-hierarchy-characterisation} in the standard
literature. It is somewhat stated informally without a proof
in~\cite{Hem89}. In order to keep this paper self-contained and for
the reader's convenience, a proof sketch of
Lemma~\ref{lem:wexp-hierarchy-characterisation} based on a proof of an
analogue characterisation of the polynomial-time hierarchy given
in~\cite{AB09} is provided in the appendix.

\subsection{Boolean Circuits} 
\label{sec:boolean-circuits}
A standard approach to raise the complexity of a problem known to be
complete for a complexity class by one exponent is to succinctly
represent the input, see \emph{e.g.}~\cite{GLV95,PY86}. A well-known
concept is to represent the input by Boolean circuits. In this paper,
for technical convenience we adapt the definition provided
in~\cite{GLV95}.

\begin{definition}
  \label{def:boolean-circuit}
  A \emph{Boolean circuit $\mathcal{C}$ of size $r$ with $n\le r$ inputs}
  is a function $f:[r]\to \{\&,\|,{\sim}, \uparrow,
  \downarrow_1\}\times [r]\times [r]$, where $f(i)=(t,j,k)$ iff the
  gate with \emph{index} $i$ is of \emph{type $t$}, \textit{i.e.}, an
  \emph{and}, \emph{or}, \emph{not}, \emph{input} or \emph{constant
    gate}, respectively, and $j,k<i$ are inputs of the gate, unless
  $t={\sim}$ in which case we require $j=k$.
\end{definition}

We identify each gate of $\mathcal{C}$ with an index from $[r]$, and
by convention the first $n\le r$ gates are $\mathit{input}$-gates, and
the $r$-th gate, \emph{i.e.}, the gate with index $r-1$, is treated as
the \emph{output gate} of $\mathcal{C}$.  Moreover for technical
convenience, we sometimes identify the various types of the gates by
natural numbers ordered as in Definition~\ref{def:boolean-circuit},
\emph{i.e.}, $\&$ is identified as $0$, $\|$ as $1$, \emph{etc.} By
using constant gates as gates with constant value $1$, an input $w\in
\{0,1\}^n$ to $\mathcal{C}$ induces a unique evaluation mapping
$e_w:[r]\to \{0,1\}$ defined in the obvious way, and $\mathcal{C}$
evaluates to true (false) on input $w$ if $e_w(r-1)=1$
($e_w(r-1)=0$). For brevity, we define $\mathcal{C}(w)\defeq
e_w(r-1)$, and if $m_1,\ldots,m_k\in \N$ then $\mathcal{C}(m_1,\ldots
m_k)$ is the output of $\mathcal{C}(w_1\cdots w_k)$, where each
$w_i\in \{0,1\}^{\lceil \log m_i \rceil}$ is the binary, if necessary
padded, representation of $m_i$.

For the remainder of this section, we will briefly recall and
elaborate on some results and concepts about circuits and succinct
encodings from Papadimitriou's book~\cite{Pap94} on computational
complexity. Given a circuit $\mathcal{C}$ and an input $w\in
\{0,1\}^n$ for some $n\ge 0$, it is well-known that determining
$\mathcal{C}(w)$ is
$\P$-complete~\cite[Thm.\ 8.1]{Pap94}. In~\cite{Pap94}, the proof of
$\P$-hardness is established by showing that the computation table of
a polynomial-time Turing machine can be encoded as a Boolean
circuit. For an $f(n)$-time-bounded Turing machine $M$, a computation
table is an $f(n)\times f(n)$ grid of cells $T_{i,j}$ from an alphabet
that allows for uniquely encoding configurations of $M$ such that the
configuration of $M$ in step $i$ while running on $w$ is encoded in
the $i$-th row. Figure~\ref{fig:computation-table} graphically
illustrates the concept of a computation table, where $0$ and $1$ are
alphabet symbols of $M$, and $\triangleright$ and $\Box$ are left
delimiters and blank symbols, respectively. The crucial fact for
encoding computation tables as Boolean circuits is that for $i,j>1$,
the symbol at $T_{i,j}$ only depends on a fixed number of cells,
namely $T_{i-1,j-1}$, $T_{i-1,j}$ and $T_{i-1,j+1}$, illustrated by
the gray-shaded cells in Figure~\ref{fig:computation-table}. It is
then clear that the alphabet of a computation table can be encoded
into a binary alphabet of truth values, and that a constant basic
circuit can be constructed from $M$ which ensures that the values of
the cells are correctly propagated along the $y$-axis. It then follows
that $M$ accepts $w$ iff there exists a computation table ending in an
accepting state iff the circuit encoding this computation table
evaluates to true.

\begin{figure}
  \begin{center}
      \begin{tabular}{|c|c|c|c|c|c|}
        \hline
        $\vdots$ & $\vdots$ & $\vdots$ & $\vdots$ & $\vdots$ & $\vdots$\\
        \hline
        $\triangleright$ & $1$ & $1$ & $0_{q_0}$ & $\Box$ & $\Box$\\
        \hline
        $\triangleright$ & $1$ & \cellcolor{Gray}$1_{q_1}$ & $1$ & $\Box$ & $\Box$\\
        \hline
        $\triangleright$ & \cellcolor{Gray}$1_{q_2}$ & \cellcolor{Gray}$0$ & \cellcolor{Gray}1 & $\Box$ & $\Box$\\
        \hline
        $\triangleright_{q_0}$ & $1$ & 0 & 1 & $\Box$ & $\Box$\\
        \hline
      \end{tabular}
  \end{center}
  \caption{Graphical illustration of a computation table of a
    time-bounded Turing machine $M$. For example, here we have
    $T_{2,2}=0_{q_2}$. The control state and the head position of $M$
    is indicated by tape symbols with some $q_i$ as subscript. The
    four gray-shaded cells illustrate that successive cells only
    depend on three preceding cells.}
  \label{fig:computation-table}
\end{figure}

\begin{figure*}[t]
  \begin{center} 
    \includegraphics[scale=1]{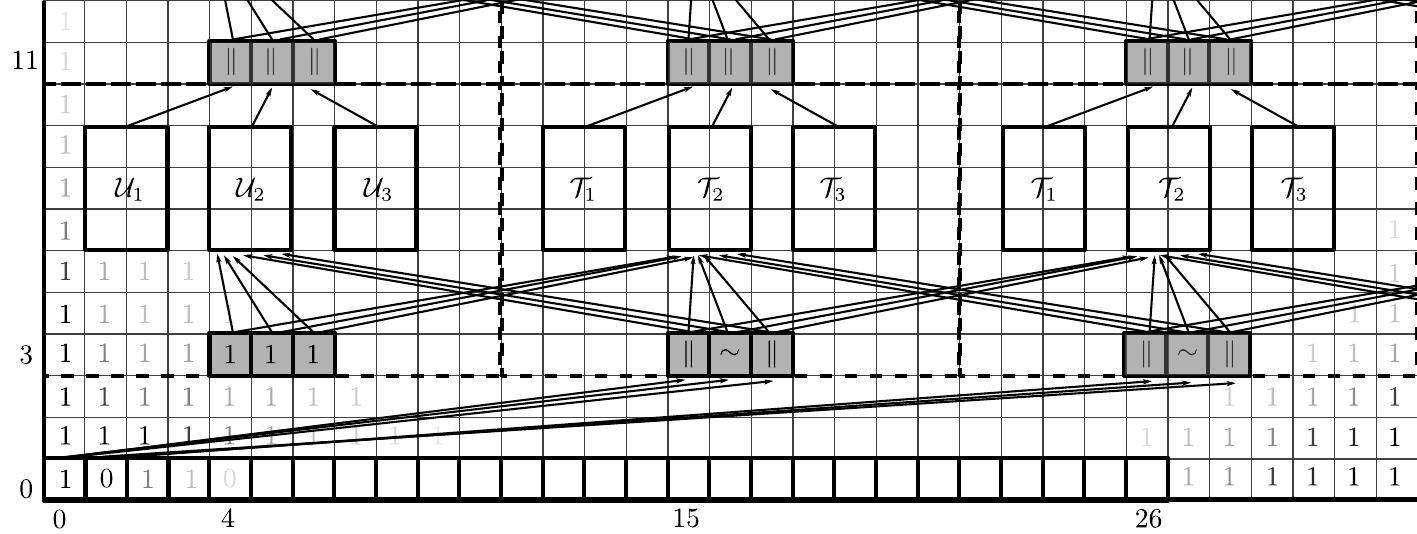}
  \end{center}
  \caption{Illustration of the approach of how to succinctly encode a
    Boolean circuit encoding the computation table of a
    polynomial-time Turing machine on an input of exponential size.
    Each square represents a gate, all gates not surrounded by boxes
    are assumed to be gates with constant value 1.}
  \label{fig:encoding}
\end{figure*}

In the next section, for our lower bound we will apply
Lemma~\ref{lem:wexp-hierarchy-characterisation}, which entails
deciding $(w,w_1,\ldots,w_i)\in R$, where the $w_j$ are of size
exponential in $n=\abs{w}$. Let $M_w$ be a polynomial-time Turing
machine deciding $R$ for a fixed $w$. The $w_1,\ldots w_j$ will
implicitly be coded into natural numbers, so it will not be possible
to construct a Boolean circuit $\mathcal{C}_w$ upfront that can
evaluate $M_w$ on the input $w_1,\ldots w_j$ of exponential size,
since we are required to establish a polynomial-time
reduction. Instead, we will \emph{succinctly encode $\mathcal{C}_w$}
via another Boolean circuit $\mathcal{D}_w$. More precisely,
$\mathcal{C}_w$ is encoded via $\mathcal{D}_w$ as follows:
$\mathcal{D}_w$ has $3r(n)+3$ input gates for some fixed polynomial
$r$ depending on $M_w$ such that for $i,j,k\in [2^{r(n)}]$ and $t\in
[5]$, $\mathcal{D}_w(t,i,j,k)=1$ iff the defining function $f$ of
$\mathcal{C}_w$ gives $f(i)=(t,j,k)$, \emph{i.e.}, that the gate with
index $i$ of $\mathcal{C}_w$ is of type $t$ and has input gates with
index $j$ and $k$. In particular, $\mathcal{D}_w$ and henceforth
$\mathcal{C}_w$ only depend on $w$ and $M$, and are \emph{independent}
of $w_1,\ldots w_i$. Note that we can view an assignment of truth
values to the gates of $\mathcal{C}_w$ as a string of length
$2^{r(n)}$.

More generally, it is known that if $\mathcal{C}$ with no input gates
is succinctly given by some circuit $\mathcal{D}$, determining whether
$\mathcal{C}$ evaluates to true is $\EXP$-complete~\cite[Thm.\ 20.2 \&
  Cor.\ 2]{Pap94}. The idea underlying the hardness proof is a
straight-forward generalisation of the approach outlined in the
paragraph above. The circuit $\mathcal{C}$ encodes the computation
table of an $\EXP$ Turing machine $M$. Since the indices of the gates
of $\mathcal{C}$ can be represented in \emph{binary}, via
$\mathcal{D}$ we can encode $\mathcal{C}$ by implicitly encoding an
exponential number of the constant basic circuit ensuring proper
propagation between consecutive cells. This approach can now be
adapted for our purpose, \emph{i.e.}, to evaluate a polynomial-time
Turing machine on an input of exponential size. The major challenge is
to transfer input to the succinctly encoded circuit on-the-fly.

Referring to Lemma~\ref{lem:wexp-hierarchy-characterisation} and given
$M_w$ as above, we can construct in logarithmic space a Boolean
circuit $\mathcal{D}_w$ encoding $\mathcal{C}_w$ such that the input
$(w_1,\ldots,w_i)$ to $\mathcal{C}_w$ is obtained from the first
$i2^{q(n)}$ gates, and $\mathcal{C}_w$ encodes a computation table of
$M_w$ on this input. Figure~\ref{fig:encoding} illustrates how this
can be realised. Each box in Figure~\ref{fig:encoding} is a gate, and
a cell of the computation table of $M_w$ while being executed on
$(w_1,\ldots,w_n)$ is encoded into the dashed boxes, or more
specifically, into the three framed gray-shaded boxes on the bottom of
the dashed boxes. Here, we assume that three bits are sufficient to
represent the alphabet of the computation table of $M_w$, and that
$\triangleright_{q_0}$ is encoded as $111$. Consequently, the gates
with index $(3,4),(3,5)$ and $(3,6)$, representing the cell $T_{1,1}$
of the computation table of $M_w$, are gates with constant value one,
as indicated in Figure~\ref{fig:encoding}. Now we want the values of
the cells $T_{1,2}$, $T_{1,3}$, \emph{etc.} of the computation table
of $M_w$ to be equivalent to $w_1\cdots w_n$, which are represented by
the gates with indices $(0,0), \ldots (0,i2^{q(n)})$. The gates in
$\mathcal{D}_w$ corresponding to $T_{1,2}$ and $T_{1,3}$ have indices
$(3,15),(3,16),(3,17)$ and $(3,26), (3,27), (3,28)$, respectively.
Those gates have the gates $(0,0)$ and $(0,1)$ as their inputs,
respectively. Suppose that in our encoding $1$ is represented as $101$
and $0$ as $010$, the sequence of $\|$, $\sim$ and $\|$ gates ensures
that $1$ is mapped to $101$ and $0$ to $010$. Consequently, the gates
$(3,15),(3,16),(3,17)$ can correctly transfer the alphabet symbols
$\{0,1\}$ of $M_w$ into the internal representation of the computation
table, and in particular copy the first symbol of the input string
$w_1\cdots w_i$ into the internal representation of the computation
table. In the example in Figure~\ref{fig:encoding}, the gates with
index $(3,15),(3,16),(3,17)$ would output $1$, $0$ and $1$,
respectively, since the gate $(0,0)$ has value $1$ which corresponds
to the first symbol of the input string $w_1$. As stated before, in
our reduction this value is provided on-the-fly. The rest of the
reduction follows standard arguments. Each dashed box contains
circuits $\mathcal{T}_1, \mathcal{T}_2$ and $\mathcal{T}_3$ which
compute the consecutive cell of the simulated computation table of
$M_w$, \emph{i.e.}, the values of the three gates representing this
cell. The dashed boxes on the left use different circuits
$\mathcal{U}_1, \mathcal{U}_2$ and $\mathcal{U}_3$ since they do not
have a left neighbor. All unused gates can assumed to be dummy gates,
\emph{i.e.} gates with constant value $1$, as indicated in
Figure~\ref{fig:encoding}. It follows that $M_w$ accepts
$(w_1,\ldots,w_i)$ iff $\mathcal{C}_w$ evaluates to true on the input
provided, \emph{i.e.}, the value of the gate with the highest index of
$\mathcal{C}_w$ is equal to 1.

In order to encode $\mathcal{C}_w$ succinctly, it is clear that due to
the regular structure of $\mathcal{C}_w$, the type and input gates to
any gate can be computed from a given index of a gate by a
polynomial-time algorithm. The circuit $\mathcal{D}_w$ can now be
taken as the circuit corresponding to this algorithm.

\section{Completeness of the $\Sigma_{i+1}$-Fragment of Presburger Arithmetic for $\phsigma_{i}^{\EXP}$}
\label{sec:complexity}

In this section, we show that PA($i+1$) is
$\phsigma_{i}^{\EXP}$-complete for every fixed $i>0$. We begin with
the lower bound and first note that it is not possible to adapt
Berman's hardness proof~\cite{Ber80} in order to get the desired
result, since it relies on a trick by Fischer \& Rabin~\cite{FR74} in
order to perform arithmetic operations on a bounded interval over
large numbers which linearly increases the number of quantifier
alternations. Instead, we will partly adapt concepts and ideas
introduced by Gr\"adel in his hardness proof for PA($2$)
in~\cite{Grae89} and Gottlob, Leone \& Veith in~\cite{GLV95}. Roughly
speaking, we aim for ``implementing''
Lemma~\ref{lem:wexp-hierarchy-characterisation} via a PA($i+1$)
formula, which will entail encoding bit strings of exponential size
into natural numbers and evaluating Boolean circuits in Presburger
arithmetic on-the-fly. The upper bound does not follow immediately and
requires combining solution intervals established by Weispfenning
in~\cite{Wei90} with Proposition~\ref{prop:frank-tardos-ip}.

\subsection{Lower Bounds}
\label{sec:lower-bound}

The goal of this section is to prove the following proposition.
\begin{proposition}
  \label{prop:pa-hardness}
  Let $L\subseteq \{0,1\}^*$ be a language in $\phsigma_{i}^{\EXP}$,
  $i>0$ and $w\in \{0,1\}^*$. There exists a polynomial-time
  computable PA($i+1$) formula $\Phi_{L,w}$ such that $w\in L$ iff
  $\Phi_{L,w}$ is valid.
\end{proposition}

To this end, we employ the characterisation of $\phsigma_i^{\EXP}$ in
Lemma~\ref{lem:wexp-hierarchy-characterisation}. Let $M$ be the
deterministic polynomial-time Turing machine deciding $R$ from
Lemma~\ref{lem:wexp-hierarchy-characterisation}, and let $M_w$ be such
a Turing machine deciding $R$ for a fixed input $w\in \{0,1\}^n$,
which can be computed from $M$ in logarithmic space. The bit strings
$w_1$ to $w_i$ from Lemma~\ref{lem:wexp-hierarchy-characterisation}
constituting the input to $M_w$ are represented in our reduction via
natural numbers assigned to first-order variables
$\vec{x}=(x_1,\ldots,x_i)$. The precise encoding of a $w_j$ via $x_j$
is discussed below. For now, it is only important to mention that not
every natural number encodes a bit string. Let us focus on the
high-level structure of $\Phi_{L,w}$:
\begin{multline}
  \label{eqn:reduction-even}
  \Phi_{L,w} \defeq \exists x_1.\forall x_2\cdots
  Q_i x_i.\bigwedge_{1\le j\le i,~j\text{
        odd}}\Psi_{\mathit{valid,r(n)}}(x_j)\wedge\\ \wedge
  \Big(\bigwedge_{1\le j\le i,~j\text{
        even}} \Psi_{\mathit{valid},r(n)}(x_j)\Big) \rightarrow
  \Psi_{M_w}(x_1,\ldots,x_i).
\end{multline}
Unsurprisingly, the alternation of quantifiers in
Lemma~\ref{lem:wexp-hierarchy-characterisation} is reflected by the
alternation of quantifiers in (\ref{eqn:reduction-even}), so
$Q_i=\exists$ if $i$ is odd and $Q_i=\forall$ if $i$ is even. The
formula $\Psi_{\mathit{valid},r(n)}(x_i)$ is a $\Pi_1$-formula, and
$\Psi_{M_w}(x_1,\ldots,x_i)$ is a formula in the Boolean closure of
$\Sigma_1$ if $i$ is odd and a $\Sigma_1$-formula if $i$ is even. The
first conjunct ensures that the existentially quantified variables
represent encodings of bit strings and the second conjunct that, under
the additional assumption that the universally quantified variables
encode valid bit strings as well, $M_w$ accepts the bit strings
encoded in $x_1,\ldots, x_i$. For the given $w\in \{0,1\}^n$, those
formulas are concrete instances of a family of formulas, and $r(n)$ is
an index in this family for some polynomial $r(n)$ which dominates
$q(n)$ in Lemma~\ref{lem:wexp-hierarchy-characterisation} and is made
more precise at a later stage. Consequently, for a fixed $i>0$, we
have that $\Phi_{L,w}$ is a PA($i+1$) formula.

In our reduction, we have to take extra care to prevent the
``accidental'' introduction of quantifier alternations. In general
when providing formulas, we adapt Gr\"adel's approach in~\cite{Grae89}
and provide \emph{neutral formulas}, which are open polynomially
equivalent $\Sigma_1$- and $\Pi_1$-formulas. This ensures that, for
instance, we do not have to care about whether we could possibly
introduce a new quantifier alternation if a formula is used on the
left-hand side of an implication. When providing a neutral
$\Sigma_1$-formula $\Phi(\vec{x})=\exists
\vec{y}.\varphi(\vec{x},\vec{y})$, we will denote by
$\bar{\Phi}(\vec{x})=\forall \vec{y}.\bar{\varphi}(\vec{x},\vec{y})$
its neutral equivalent $\Pi_1$ counterpart. For the sake of consistent
naming, whenever $\Phi(\vec{x})$ occurs as a subformula in some other
formula, we \emph{implicitly} assume that it is appropriately replaced
such that the resulting formula is either a $\Sigma_1$- or a
$\Pi_1$-formula, depending on the context. Likewise, if for instance
$\Phi(\vec{x})$ occurs as a negated subformula in a formula that is
supposed to be existentially quantified, we assume that this
subformula is \emph{implicitly} replaced by $\exists
\vec{x}.\neg(\bar{\varphi}(\vec{x},\vec{y}))$, and
$\neg(\bar{\varphi}(\vec{x},\vec{y}))$ is treated in the same way if
it is not yet quantifier-free. In this way, we can always make sure to
result in $\Sigma_1$- or $\Pi_1$-formulas.

We now turn towards the details of our reduction and begin with
discussing the encoding of bit strings as natural numbers we use
subsequently. The encoding we use is due to Gr\"adel~\cite{Grae89}. In
his $\NEXP$ lower bound for PA($2$) he exploits a result due to
Ingham~\cite{Ing40,Che10} that for any sufficiently
large\footnote{Cheng~\cite{Che10} provides explicit bounds on Ingham's
  result~\cite{Ing40} and shows that this statement holds for all
  $i\in \N$ such that $i>2^{2^{15}}$. As in~\cite{Grae89}, for brevity
  we will use Ingham's result as if it were true for all $i>0$. It
  will be clear that we could add Cheng's offset to all numbers
  involved, causing a constant blowup only.} $i\in \N$ there is at
least one prime in the interval $[i^3,(i+1)^3)$. Given a bit string
  $w=b_1\cdots b_n\in \{0,1\}^n$, a natural number $a\in \N$ encodes
  $w$ if for all $1\le i\le n$ and
\begin{align*}
  \text{for all primes } p\in [i^3,(i+1)^3): a\equiv b_i \bmod p.
\end{align*}
The existence of such an $a$ is then guaranteed by the Chinese
remainder theorem. Given a fixed $n>0$, we call $a\in \N$ a
\emph{valid encoding} if for every $1\le i\le n$, either $a\equiv
0\bmod p$ or $a\equiv 1\bmod p$ for all prime numbers $p\in
[i^3,(i+1)^3)$.
    
In order to enable the extraction of bits of bit strings encoded as
naturals, we show how to check for divisibility with a natural number
whose number of bits is fixed. Next, we show how to evaluate a Boolean
circuit in Presburger arithmetic. This serves two purposes: first, it
allows for deciding if a given number lies in an interval
$[i^3,(i+3)^3)$ and for testing whether a given number is a prime due
  to the AKS primality test~\cite{AKS02}. Second, it allows for
  simulating $M_w$ discussed above on an input of exponential size
  using its succinct encoding via a circuit as discussed in
  Section~\ref{sec:boolean-circuits}. Putting everything together
  eventually yields the desired reduction.

We begin with a family of quantifier-free formulas
$\Phi_{\mathit{bin},n}(\vec{x},x)$ such that given $\vec{b}\in
\{0,1\}^n$ and $b\in \N$, $\Phi_{\mathit{bin},n}(\vec{b},b)$ holds if
$\vec{b}$ is the binary representation of $b$. Consequently, this
formula implicitly constraints $b$ such that $b\in [2^n]$:
\begin{align}
  \Phi_{\mathit{bin},n}(\vec{x},x) & \defeq \bigwedge_{i\in [n]} (x_i=0 \vee x_i= 1) \wedge x=\sum_{i\in [n]}2^ix_i.
\end{align}
Next, we provide a family of neutral formulas
$\Phi_{\mathit{mod},n}(x,y)$ such that for $a,b$ with $b\in [2^n]$,
$\Phi_{\mathit{mod},n}(a,b)$ holds iff $a\equiv 0\bmod
b$. Essentially, $\Phi_{\mathit{mod},n}(x,y)$ realises a formula for
bounded multiplication. In contrast to a formula with the same purpose
given in~\cite{Grae89}, it is not recursively defined and of size
$O(n)$ as opposed to $O(n\log n)$ when binary encoding of numbers is
assumed. The latter fact will be useful in
Section~\ref{sec:derived-results}. The underlying idea of the
subsequent definitions is that if the binary expansion of $b$ is $b =
\sum_{i\in [n]} 2^i b_i~\text{and}~bk=a$ for some $k\ge 0$ then $a$
can be written as $a=\sum_{i\in [n]}2^i a_i$ with $a_i=kb_i$:
\begin{align}
  \notag
  \Phi_{\mathit{dig},n}(\vec{x},\vec{y},k) & \defeq \bigwedge_{i\in [n]} 
  (y_i=0\rightarrow x_i=0 \wedge  y_i=1 \rightarrow x_i=k)\\
  \notag
  \Phi_{\mathit{mod},n}(x,y) & \defeq \exists \vec{x}. \exists \vec{y}.
  \exists k. \Phi_{\mathit{bin},n}(\vec{y},y) \wedge\\& ~~~~~~~~
  \label{eqn:qfpa-mod}
  \wedge \Phi_{\mathit{dig},n}(\vec{x},\vec{y},k) \wedge x=\sum_{i\in [n]} 2^i x_i\\
  \notag
  \bar{\Phi}_{\mathit{mod},n}(x,y) & \defeq \forall \vec{x}. \forall \vec{y}.
  \forall k. \big(\Phi_{\mathit{bin},n}(\vec{y},y) \wedge\\& ~~~~~~
  \notag
  \wedge \Phi_{\mathit{dig},n}(\vec{x},\vec{y},k)\big) \rightarrow x=\sum_{i\in [n]} 2^i x_i.
\end{align}

We now turn towards evaluating Boolean circuits with suitable formulas
in Presburger arithmetic. The subsequent formulas for evaluating a
circuit $\mathcal{C}$ with $n$ input and $r$ gates in total are
essentially an adaption of a construction given by Gottlob, Leone \&
Veith in~\cite{GLV95}. It is easily checked that for $a\in [2^n]$,
$\Phi_{\mathcal{C}}(a)$ holds iff $\mathcal{C}(a)=1$. In
$\Phi_{\mathcal{C}}$, the input to $\mathcal{C}$ is encoded via a
dimension $n$ vector of first-order variables $\vec{x}$ and the
Boolean assignment to the gates via a dimension $r$ vector of
first-order variables $\vec{y}$, which are implicitly assumed to range
over $\{0,1\}$. First, we provide a formula ensuring that the
structure of the gates of $\mathcal{C}$ is correctly encoded in
$\vec{y}$:
\begin{multline}
  \label{eqn:pa-circuit-formula}
  \Phi_{\mathcal{C},\mathit{gates}}(\vec{x},\vec{y}) \defeq\\
     \bigwedge_{i\in [r]} \left \{
  \begin{array}{ll}
    y_i = 1 \leftrightarrow (y_j=1 \wedge y_k=1) & \text{if } f(i)=(\&,j,k)\\
    y_i = 1 \leftrightarrow (y_j=1 \vee y_k=1) & \text{if } f(i)=(\|,j,k)\\
    y_i = 1 \leftrightarrow y_j=0 & \text{if } f(i)=({\sim},j,k)\\
    y_i = x_i & \text{if } f(i)=(\uparrow,0,0)\\
    y_i = 1 & \text{if } f(i)=(\downarrow_1,0,0).
  \end{array}\right.
\end{multline}
Next, the formula $\Phi_\mathcal{C}(x)$ defined below now enables us
to determine whether $\mathcal{C}$ accepts a given input encoded into
the first-order variable $x$:
\begin{align}
  \label{eqn:pa-circuit-eval-exist}
  \Phi_{\mathcal{C}}(x) & \defeq \begin{aligned}[t] \exists \vec{x}.\exists \vec{y}.\exists y.
  \Phi_{\mathit{bin},n}(\vec{x},x) \wedge \Phi_{\mathit{bin},r}(\vec{y},y) 
  \wedge\\
  \wedge \Phi_{\mathcal{C},\mathit{gates}}(\vec{x},\vec{y})\wedge y_r = 1\end{aligned}\\
  \label{eqn:pa-circuit-eval-forall}
  \bar{\Phi}_{\mathcal{C}}(x) & \defeq \begin{aligned}[t]\forall \vec{x}.\forall \vec{y}.\forall y.
  \big(\Phi_{\mathit{bin},n}(\vec{x},x) \wedge \Phi_{\mathit{bin},r}(\vec{y},y) 
  \wedge\\
  \wedge \Phi_{\mathcal{C},\mathit{gates}}(\vec{x},\vec{y})\big)\rightarrow y_r = 1.\end{aligned}
\end{align}

We now show how a predicate determining whether a given number $a\in
\N$ is a prime number in the interval $[b^3,(b+1)^3)$ for some $b>0$
  representable by $n$ bits can be realised. It is easily verified
  that any number in this interval can be represented by at most
  $m=3(n+1)$ bits. Moreover as discussed above, both conditions can be
  decided in polynomial time. Therefore we can construct in
  logarithmic space a Boolean circuit $\mathcal{C}_{\mathit{prime},n}$
  with $m+n$ input gates implementing this
  predicate~\cite[Thm.\ 8.1]{Pap94} and define
\begin{multline}
  \label{eqn:prime-number-in-interval-exist}
  \Phi_{\mathit{prime},n}(x,y) \defeq \exists \vec{x}.\exists \vec{y}.\exists z.
  \Phi_{\mathit{bin},m}(\vec{x},x) \wedge
  \Phi_{\mathit{bin},n}(\vec{y},y) \wedge\\ 
  \wedge z = \sum_{i\in [m]} 2^ix_i + 2^{m}\sum_{i \in [n]} 2^i y_i \wedge
  \Phi_{\mathcal{C}_{\mathit{prime},n}}(z)
\end{multline}
\begin{multline}
  \label{eqn:prime-number-in-interval-forall}
  \bar{\Phi}_{\mathit{prime},n}(x,y) \defeq \forall \vec{x}.\forall \vec{y}.\forall z.\big(
  \Phi_{\mathit{bin},m}(\vec{x},x) \wedge
  \Phi_{\mathit{bin},n}(\vec{y},y) \wedge\\ 
  \wedge z = \sum_{i\in [m]} 2^ix_i + 2^{m} \sum_{i \in [n]} 2^i y_i\big) 
  \rightarrow \Phi_{\mathcal{C}_{\mathit{prime},n}}(z).
\end{multline}
The first line of $\Phi_{\mathit{prime},n}(x,y)$ converts $x$ and $y$
into their binary representation. Next, the second line first
concatenates these bit representations via the additional variable $z$
by appropriately shifting the value of $y$ by $m$ bits, and finally
$z$ is passed to $\mathcal{C}_{\mathit{prime},n}$. Consequently, we
have that $\Phi_{\mathit{prime},n}(a,b)$ holds iff $a$ is prime and
$a\in [b^3,(b+1)^3)$.

We are now in a position in which we can define a family of
$\Pi_1$-formulas $\Psi_{\mathit{valid},n}(x)$ used in
(\ref{eqn:reduction-even}) that allow for testing whether some $a\in
\N$ represents a valid respectively invalid encoding of a bit string
of length $2^n$. For valid encodings, we wish to make sure that all
primes in every relevant interval $[b^3,(b+1)^3)$ have uniform residue
  classes in $a$ for all $1\le b\le 2^n$, \emph{i.e.}, for any two
  primes $p_1,p_2\in [b^3,(b+1)^3)$ we either have $a\equiv 0\mod p_1$
    and $a\equiv 0\mod p_2$, or $a-1\equiv 0\mod p_1$ and $a-1\equiv
    0\mod p_2$. Let $m$ be as above,
\begin{multline*}
  \Psi_{\mathit{valid},n}(x) \defeq \forall y.\forall p_1.\forall
  p_2.\big(1\le y\le 2^{n} \wedge 
  \\\wedge \Phi_{\mathit{prime},{n+1}}(p_1,y) \wedge \Phi_{\mathit{prime},{n+1}}(p_2,y)\big)
  \rightarrow\\ \rightarrow \big((\Phi_{\mathit{mod},m+3}(x,p_1) \wedge 
  \Phi_{\mathit{mod},m+3}(x,p_2)) 
  \vee\\\vee (\Phi_{\mathit{mod},m+3}(x-1,p_1) \wedge \Phi_{\mathit{mod},m+3}(x-1,p_2))\big).
\end{multline*}

In order to complete our hardness proof for $\phsigma_i^\EXP$ for a
subsequently fixed $i>0$ via its characterisation in
Lemma~\ref{lem:wexp-hierarchy-characterisation} and $\Phi_{L,w}$
in~(\ref{eqn:reduction-even}), we will now define the remaining
$\Pi_1$-formula $\Psi_{M_w}(x_1,\ldots,x_i)$ for a given $w\in
\{0,1\}^n$. Let $\mathcal{C}_w$ be the Boolean circuit succinctly
encoded by a Boolean circuit $\mathcal{D}_w(t,y,z_1,z_2)$ deciding
$M_w$ on an input of length $2^{q(n)}i$ such that $\mathcal{C}_w$
consists of $2^{r(n)}$ gates for some polynomial $r:\mathbb{N}\to
\mathbb{N}$. Recall that we can view an assignment of truth values to
the gates of the succinctly encoded circuit $\mathcal{C}_w$ as a bit
string, or sequence of bit strings, of appropriate length. In the
following let $\vec{a}=(a_1,\ldots,a_i)\in \N^{i}$ be a valuation, for
any $1\le j< i$ each $a_j$ will be used to encode the values of the
input gates with index $2^{q(n)}(j-1)$ up to $2^{q(n)}j-1$ of
$\mathcal{C}_w$, and $a_i$ will encode the values of the gates with
index $2^{q(n)}(i-1)$ up to the gate with index $2^{r(n)}-1$ of
$\mathcal{C}_w$. So in particular the internal gates of
$\mathcal{C}_w$ are encoded in $a_i$.

In order to extract encodings of bit strings from natural numbers, as
a first step we provide neutral formulas $\Phi_{\mathcal{C}_w,0}(x,y)$
and $\bar{\Phi}_{\mathcal{C}_w,0}(x,y)$ which assume $x$ to be a valid
encoding. These formulas enable us to test whether a bit of a bit
string whose index is given by $y$ is encoded to be zero in $x$.
Formally, for a valid encoding $a\in \N$ and for $b\in [2^{r(n)}]$, we
have $\Phi_{\mathcal{C}_w,0}(a,b)$ iff there is a prime $p\in
[(b+1)^3,(b+2)^3)$ and $a\equiv 0\bmod p$, or $a\equiv 0 \bmod p$ for
  all primes $p\in [(b+1)^3,(b+2)^3)$, respectively\footnote{In order
      to properly handle the case $b=0$, we have to shift the interval
      we use for the encoding by one from $[b^3,(b+1)^3)$ to
        $[(b+1)^3,(b+2)^3)$.}. Let $r'(x)=r(x)+1$, we define:
\begin{align*}
  \Phi_{\mathcal{C}_w,0}(x,y) & \defeq \exists
  p.\Phi_{prime,r'(n)}(p,y+1) \wedge \Phi_{\mathit{mod},r'(n)}(x,p)\\ 
  \bar{\Phi}_{\mathcal{C}_w,0}(x,y) & \defeq \forall
  p.\Phi_{prime,r'(n)}(p,y+1) \rightarrow \Phi_{\mathit{mod},r'(n)}(x,p).
\end{align*}
The formulas $\Phi_{\mathcal{C}_w,1}(x,y)$ and
$\bar{\Phi}_{\mathcal{C}_w,1}(x,y)$ testing whether the bit with index
$y$ is set to 1 in the encoding $x$ can be defined analogously by
negating $\Phi_{\mathcal{C}_w,0}(x,y)$. The previously constructed
formulas now enable us to define formulas that allow for evaluating
the succinctly encoded $\mathcal{C}_w$ on an input that is provided
on-the-fly via $\vec{a}$. Given an index $b$ implicitly less than
$2^{r(n)}$ of a gate of $\mathcal{C}_w$, represented by the
first-order variable $y$, and a vector of valid encodings
$\vec{a}=(a_1,\ldots,a_i)$ represented by the first-order variables
$\vec{x}$, the following formula
$\Phi_{\mathcal{C}_w,\top}(\vec{x},y)$ checks whether the value of the
gate with index $b$ is set to true under the valuation $\vec{a}$
according to the convention described before:
\begin{multline*}
  \Phi_{\mathcal{C}_w,\top}(\vec{x},y) \defeq \bigwedge_{1\le j< i}\Big(\big(2^{q(n)}(j-1) \le y
  < 2^{q(n)}j \rightarrow\\\rightarrow \Phi_{\mathcal{C}_w,1}(x_j,y)\big)
  \wedge \big(2^{q(n)}(i-1) \le y \rightarrow \Phi_{\mathcal{C}_w,1}(x_i,y)\big)\Big).
\end{multline*}
A formula $\Phi_{\mathcal{C}_w,\bot}(\vec{x},y)$ testing whether the
value of a gate is set to false can be defined analogously by negating
$\Phi_{\mathcal{C}_w,\top}(\vec{x},y)$. Building upon those formulas,
we can now construct Boolean connectives that allow for checking that
the gates of $\mathcal{C}_w$ are consistently encoded. Given
$\vec{a}\in \mathbb{N}^i$ as above,
$\Phi_{\mathcal{C}_w,\&}(\vec{a},b,c_1,c_2)$ holds if the logical
and-connective holds for the truth values of the gates with index $b,
c_1$ and $c_2$ encoded via $\vec{a}$:
\begin{multline*}
  \Phi_{\mathcal{C}_w,\&}(\vec{x},y,z_1,z_2) =\\
  \big(\Phi_{\mathcal{C}_w,\top}(\vec{x},y)\leftrightarrow
  \Phi_{\mathcal{C}_w,\top}(\vec{x},z_1) \wedge
  \Phi_{\mathcal{C}_w,\top}(\vec{x},z_2)\big).
\end{multline*}
The remaining Boolean connectives found in
Definition~\ref{def:boolean-circuit} can be reflected via the
additional formulas
\begin{align*}
  \Phi_{\mathcal{C}_w,\|}(\vec{x},y,z_1,z_2),~\Phi_{\mathcal{C}_w,\sim}(\vec{x},y,z_1,z_2)
  \text{ and } \Phi_{\mathcal{C}_w,\downarrow_1}(\vec{x},y,z_1,z_2)
\end{align*}
which are defined analogously to
$\Phi_{\mathcal{C}_w,\&}(\vec{x},y,z_1,z_2)$.  These formulas now
enable us to define a $\Pi_1$-analogue to
$\Phi_{\mathcal{C},\mathit{gates}}(\vec{x},\vec{y})$ in
(\ref{eqn:pa-circuit-formula}) in order to check if the Boolean
assignment of the succinctly encoded circuit $\mathcal{C}_w$ is
consistent:
\begin{multline*}
  \Psi_{\mathcal{C}_w,\mathit{gates}}(\vec{x}) = 
  \forall t.\forall y.\forall z_1.\forall z_2.
  \Phi_{\mathcal{D}_w}(t,y,z_1,z_2) \rightarrow \\ \rightarrow\bigwedge
  \left\{
  \begin{array}{l}
    t = 0 \rightarrow \Phi_{\mathcal{C}_w,\&}(\vec{x},y,z_1,z_2)\\
    t = 1 \rightarrow \Phi_{\mathcal{C}_w,\|}(\vec{x},y,z_1,z_2)\\
    t = 2 \rightarrow \Phi_{\mathcal{C}_w,\top}(\vec{x},y) \leftrightarrow 
    \Phi_{\mathcal{C}_w,\bot}(\vec{x},z_1)\\
    t = 4 \rightarrow \Phi_{\mathcal{C}_w,\top} (\vec{x}).
  \end{array}
  \right.
\end{multline*}
Here, $\Phi_{\mathcal{D}_w}(t,y,z_1,z_2)$ is an instantiation of
$\Phi_{\mathcal{C}}(x)$ defined in (\ref{eqn:pa-circuit-eval-exist})
and (\ref{eqn:pa-circuit-eval-forall}) for the circuit
$\mathcal{D}_w$. The use of $\Phi_{\mathcal{D}_w}(t,y,z_1,z_2)$ is not
totally clean as $\Phi_{\mathcal{C}}(x)$ it is only open in
$x$. However, this can easily be fixed by concatenating $t, y, z_1$
and $z_2$ into a single first-order variable as it was done in
(\ref{eqn:prime-number-in-interval-exist}) and
(\ref{eqn:prime-number-in-interval-forall}), and details have only
been omitted for the sake of readability. Also note that the values of
$t, y, z_1$ and $z_2$ are then implicitly bounded through
$\Phi_{\mathcal{D}_w}(t,y,z_1,z_2)$.

Finally, we can define $\Psi_{M_w}(x_1,\ldots,x_i)$, the last
remaining formula from~(\ref{eqn:reduction-even}), as follows
\begin{multline*}
  \Psi_{M_w}(x_1,\ldots,x_i) \defeq \\ 
  \left\{
  \begin{array}{ll}
    \Psi_{\mathcal{C}_w,\mathit{gates}}(\vec{x}) \wedge
    \Phi_{\mathcal{C}_w,\top}(\vec{x}, 2^{r(n)}-1) & \text{ if } i \text{ is odd}\\
    \Psi_{\mathcal{C}_w,\mathit{gates}}(\vec{x}) \rightarrow
    \Phi_{\mathcal{C}_w,\top}(\vec{x}, 2^{r(n)}-1) & \text{ if } i \text{ is even}
  \end{array}
  \right.
\end{multline*}
Inspecting the construction outlined in this section, it is not
difficult to see that for a given $w\in \{0,1\}^n$ the construction of
$\Phi_{L_w}(x_1,\ldots,x_i)$ is tedious, but can be performed in
polynomial time with respect to $n$. We leave it as an open problem
whether this reduction can actually be performed in logarithmic space,
though there do not seem to be any major obstacles. Following the
argumentation of this section, we conclude that $\Phi_{L,w}$ is the
formula required in Proposition~\ref{prop:pa-hardness}.

\subsection{Upper Bounds}
\label{sec:upper-bound}

We will now show that the previously obtained lower bounds have
corresponding upper bounds. Let us first recall an improved version of
a result by Reddy \& Loveland~\cite{RL78} established by
Weispfenning~\cite[Thm.~2.2]{Wei90}, which bounds the solution
intervals of Presburger formulas.
\begin{proposition}[Weispfenning~\cite{Wei90}]
  \label{prop:weispfenning-bounds}
  There exists a constant $c>0$ such that for any PA($i$,$j$) formula
  $\Phi$ and $N=\{0,\ldots 2^{c|\Phi|^{(3j)^i}}\}$, $\Phi$ is valid
  iff $\Phi$ is valid when restricting the first-order variables of
  $\Phi$ to be interpreted over elements from $N$.
\end{proposition}
Together with Lemma~\ref{lem:wexp-hierarchy-characterisation}, this
immediately gives that for any fixed $i>0$, validity in PA($i+1$) is
in $\phsigma_{i+1}^{\EXP}$. We now show how to decrease the number of
oracle calls by one.

To this end, let $\Phi$ be a PA($i+1$,$j$) formula in prenex normal
form for a fixed $i>0$ and some $j$, \ie,
\begin{align*}
  \Phi=\exists \vec{x}_1.\forall \vec{x}_2\cdots
  Q_{i+1}\vec{x}_{i+1}.\varphi(\vec{x}_1,\ldots,\vec{x}_{i+1}). 
\end{align*}
In order to decide validity of $\Phi$, by application of
Proposition~\ref{prop:weispfenning-bounds}, a
$\phsigma_i^{\EXP}$-algorithm can alternatingly guess valuations
$\vec{a}_1,\ldots \vec{a}_i\in \N^j$ for the $\vec{x}_1,\ldots
\vec{x}_i$ such that $\norm{\vec{a}_k}\le
2^{c\abs{\Phi}^{(3j)^{i+1}}}$ for all $1\le k\le i$ and some constant
$c>0$, and by additionally padding valuations with leading zeros, we
may assume that any number in every $\vec{a}_k$ is represented using
$2^{\poly(\abs{\Phi})}$ bits. Consequently, it remains to show that
validity of $Q_{i+1} \vec{x}_{i+1}.\varphi(\vec{a}_1/\vec{x}_1,\ldots,
\vec{a}_i/\vec{x}_i, \vec{x}_{i+1})$ can be decided in polynomial
time. This is, of course, not the case under standard assumptions from
complexity theory. However, the final call to the
$\phsigma_0^\P$-oracle of a $\phsigma_i^{\EXP}$ algorithm gets
$\varphi$ and all $\vec{a}_k$ as input, the latter being of
\emph{exponential size} in $\abs{\Phi}$. Informally speaking, this
provides us with sufficient additional time in order to decide
validity of
\begin{align}
  \label{eqn:final-call}
  Q_{i+1}\vec{x}_{i+1}.\varphi(\vec{a}_1/\vec{x}_1,\ldots,
  \vec{a}_i/\vec{x}_i, \vec{x}_{i+1})
\end{align}
in polynomial time with respect to the size of the input.

\begin{algorithm}[tb]
\begin{algorithmic}[1]
\STATE $\varphi(\vec{x}_1,\ldots,\vec{x}_{i+1}) := $DNF$(\varphi(\vec{x}_1,\ldots,\vec{x}_{i+1}))$
\FORALL{clauses $\psi(\vec{x}_1,\ldots,\vec{x}_{i+1})$ of $\varphi$}
\FORALL{literals $t=p(\vec{x}_1,\ldots,\vec{x}_{i+1})< b$ of $\psi$}
\STATE replace $t$ in $\psi$ with 
$-p(\vec{x}_1,\ldots,\vec{x}_{i+1}) \ge -b+1$
\ENDFOR
\FORALL{literals $t=\neg(p(\vec{x}_1,\ldots,\vec{x}_{i+1})<b)$}
\STATE replace $t$ in $\psi(\vec{x}_1,\ldots,\vec{x}_{i+1})$ with 
$p(\vec{x}_1,\ldots,\vec{x}_{i+1})\ge b$
\ENDFOR
\STATE $(S : A\vec{x}_{i+1} \ge \vec{c}) := \psi[\vec{a}_1/\vec{x}_1,\ldots,\vec{a}_i/\vec{x}_i]$
\IF {$\eval{S}\neq \emptyset$}
\RETURN true
\ENDIF 
\ENDFOR
\RETURN false
\end{algorithmic}
\caption{Deciding $\exists
  \vec{x}_{i+1}.\varphi(\vec{x}_1,\ldots,\vec{x}_{i+1})$ for a given
  instantiation $\vec{a}_1,\ldots \vec{a}_i\in \N^j$ of
  $\vec{x}_1,\ldots \vec{x}_i$.}
\label{alg:satisfiability}
\end{algorithm}

Algorithm~\ref{alg:satisfiability}, which takes $\varphi$ and the
$\vec{a}_k$ as input, is a pseudo algorithm deciding validity of a
formula as in (\ref{eqn:final-call}) for even $i$, \ie,
$Q_{i+1}=\exists$. The case $Q_{i+1}=\forall$ can be derived
symmetrically. Let us discuss Algorithm~\ref{alg:satisfiability} and
analyse its running time. In Line~1, the algorithm converts $\varphi$
into disjunctive normal form.  This step can be performed in
exponential time $\DTIME(2^{O(\abs{\Phi})})$ and thus takes polynomial
time with respect to the input. Starting in Line~2, the algorithm
iterates over all clauses $\psi$ of $\varphi$, and since there are at
most $2^{O(\abs{\Phi})}$ clauses this iteration is performed at most a
polynomial number of times with respect to the size of the input. In
each iteration, in Lines~3--8 the algorithm transforms the disjuncts
of $\psi$ into linear inequalities by eliminating negation. After
Line~8, $\psi$ is a conjunction of linear inequalities and thus gives
rise to an equivalent system of linear Diophantine inequalities $S$ in
which the first-order variables $\vec{x}_1,\ldots \vec{x}_i$ are
instantiated by the $\vec{a}_1,\ldots \vec{a}_i$. Clearly, Lines~3--9
can be executed in polynomial time with respect to the size of the
input. Finally, in Line~10 feasibility of $S$ is checked. To this end,
we invoke Proposition~\ref{prop:frank-tardos-ip} from which it follows
that feasibility of each $S$ can be decided in
$\DTIME(2^{p(\abs{\Phi})}\abs{S})$ for some polynomial $p$. This step
is again polynomial with respect to the input to the oracle call. If
$S$ is feasible the algorithm returns true in Line~11. Otherwise, if
no $S$ is feasible for all disjuncts of $\varphi$, the algorithm
returns false in Line~14. Consequently, we have shown the following
proposition, which together with Proposition~\ref{prop:pa-hardness}
completes the proof of Theorem~\ref{thm:main}.

\begin{proposition}
  \label{prop:pa-upper-bound}
  For any fixed $i>0$, PA($i+1$) is decidable in $\phsigma_i^{\EXP}$.
\end{proposition}

\subsection{Discussion}
We conclude this part of the paper with a short discussion on the
relationship of our proof of the lower bound of PA(${i+1}$) to the
proof of a $\NEXP$ lower bound for PA($2$) by Gr\"adel~\cite{Grae89},
and applications of and results derivable from
Proposition~\ref{prop:pa-upper-bound}.

As it emerged in Section~\ref{sec:lower-bound}, at many places we can
apply and reuse ideas of Gr\"adel's $\NEXP$-hardness proof for PA($2$)
given in~\cite{Grae89} for our lower bound. One main difference is
that for his hardness proof, Gr\"adel reduces from a $\NEXP$-complete
tiling problem that he specifically introduces in order to show
hardness for PA(2). In our paper, we are in the lucky position of
having access to twenty-five additional years of developments in
computational complexity, in which it turned out that succinct
encodings via Boolean circuits provide a canonical way in order to
show hardness results for complexity classes that include $\EXP$, see
\eg~\cite{PY86,Pap94,GLV95}. Moreover, the discovery of a
polynomial-time algorithm for deciding primality~\cite{AKS02} also
enables us to use Boolean circuits encoded into $\Sigma_1$-
respectively $\Pi_i$-formulas in order to decide primality of a
positive integer of a bounded bit size, while in~\cite{Grae89} this is
achieved by an application of the Lucas primality criterion,
\emph{cf.}~ Lehmer's more general proof~\cite{Leh27}. In addition,
Gr\"adel's stronger statement that PA($2$) is $\NEXP$-hard already for
an $\exists \forall^*$-quantifier prefix can be recovered from our
lower bound. Even more generally for $i>1$, we can derive
$\phsigma_i^{\EXP}$-hardness from our construction for a
$(\exists\forall)^{((i-1)/2)}\exists^*\forall^*$ quantifier prefix if
$i$ is odd, and for a
$\exists(\forall\exists)^{(i/2-1)}\forall^*\exists^*$ quantifier
prefix if $i$ is even. Even though essentially all technical results
required to prove Theorem~\ref{thm:main} were available
when~\cite{Grae89} was published, as we have seen in this section the
proof of the lower bound requires some substantial technical efforts,
which is probably a reason why this result has not been obtained
earlier.

With regards to applications of Proposition~\ref{prop:pa-upper-bound},
we give an example of a result which can be obtained as a corollary of
this proposition. In~\cite{Huy85}, Huynh investigates the complexity
of the inclusion problem for context-free commutative grammars. Given
context-free grammars $G_1,G_2$, this problem is to determine whether
the Parikh image\footnote{The Parikh image of a word $w\in \Sigma^*$
  is a vector of naturals of dimension $|\Sigma|$ counting the number
  of times each alphabet symbol occurs in $w$.} of the language
defined by $G_1$ is included in the Parikh image of the language
defined by $G_2$. Building upon a careful analysis of the semi-linear
sets obtained from Parikh images of context-free grammars due to
Ginsburg~\cite{Gin66} and by establishing a Carath{\'e}odory-type
theorem for integer cones, Huynh shows that the complement of this
problem is in $\NEXP$. This result can however now easily be obtained
as a corollary of Proposition~\ref{prop:pa-upper-bound}: Verma
\emph{et al.} have shown that the Parikh image of a context-free
grammar can be defined in terms of a $\Sigma_1$-formula of Presburger
arithmetic linear in the size of the
grammar~\cite{VSS05}. Non-inclusion then reduces to checking validity
of a $\Sigma_2$-sentence, which yields the following corollary.
\begin{corollary}
  \label{cor:ccfg-inclusion}
  Non-inclusion between Parikh images of context-free grammars is in
  $\NEXP$.
\end{corollary}
Of course, the ``hard work'' of the upper bound is done in
Proposition~\ref{prop:frank-tardos-ip}, but nevertheless we are able
to obtain a succinct proof of Huynh's result. In general, the $\NEXP$
upper bound for PA(2) provides a generic upper bound for non-inclusion
problems that can be reduced to checking inclusion between semi-linear
sets definable via PA(1) formulas. For context-free commutative
grammars, it should however be noted that it is not known whether this
upper bound is tight, the best known lower bound being
$\phsigma_2^\P$~\cite{Huy85}.

\section{Ultimately-Periodic Sets Definable in the $\Sigma_1$-fragment
  of Presburger Arithmetic}
\label{sec:derived-results}

We will now apply some techniques developed in
Section~\ref{sec:complexity} in order to prove
Theorem~\ref{thm:qfpa-projections}, \emph{i.e.}, give bounds on the
representation of projections of PA($1$) formulas open in one variable
as ultimately-periodic sets. Formally, given a PA($1$) formula
$\Phi(x)$, we are interested in the representation of the set
\begin{align*}
  \eval{\Phi(x)}=\{ a\in \N : \Phi(a/x) \text{ is valid}\}. 
\end{align*}
Subsequently, we show that this set is an ultimately periodic set
whose period is at most doubly-exponential and that this bound is
tight. Throughout this section we assume binary encoding of numbers in
$\Phi(\vec{x})$

We begin with the first part of Theorem~\ref{thm:qfpa-projections} and
prove the following proposition.
\begin{proposition}
  There exists a family of $\Sigma_1$-formulas of Presburger
  arithmetic $(\Phi_n(x))_{n>0}$ such that each $\Phi_n(x)$ is a
  {PA($1$,$O(n)$)} formula with $\abs{\Phi_n(x)}\in O(n^2)$ and
  $\eval{\Phi_n(x)}$ is an ultimately periodic set with period $p_n\in
  2^{2^{\Omega(n)}}$.
\end{proposition}
To this end, we combine $\Phi_{\mathit{mod},n}(x,y)$ from
(\ref{eqn:qfpa-mod}) in Section~\ref{sec:lower-bound} with the
following statement.
\begin{proposition}[Nair~\cite{Nai82}]
  \label{prop:nair}
  Let $n\ge 9$, then $2^n\le \lcm\{1,\ldots n\}\le 2^{2n}$.
\end{proposition}
We define
\begin{align*}
  \Phi_n(x) & \defeq \exists y.\Phi_{\mathit{mod},n}(x,y) \wedge y > 1.
\end{align*}
We have $\abs{\Phi_{n}(x)}\in \bigo(n^2)$, and, since numbers are
encoded in binary, that $\Phi_n(x)$ is a PA($1$,$O(n)$) formula. Now
$a\in \eval{\Phi(x)}$ iff there is $1< m<2^n$ such that $a \equiv 0
\bmod m$, and consequently
\begin{align*}
  \eval{\Phi_{n}(x)} &= \bigcup_{1<m<2^n} U(0,m,\emptyset,\{0\})\\
  &= U(0,p,\emptyset,\{ a : a\in [p], m|a, 1<m<2^n \}),
\end{align*}
where $p\defeq \lcm\{1,\ldots 2^n-1\}$. By
Proposition~\ref{prop:nair}, $p\in 2^{2^{\Omega(n)}}$, which yields
the lower bound for Theorem~\ref{thm:qfpa-projections}.

Turning now towards the upper bound, the remainder of this section is
devoted to proving the second part of
Theorem~\ref{thm:qfpa-projections}, \ie, the following statement.

\begin{proposition}
  For any $\Sigma_1$-formula $\Phi(x)$,
  we have $\eval{\Phi(x)}=U(t,p,B,R)$ such that $t\in
  2^{\poly(\abs{\Phi(x)})}$ and $p\in 2^{2^{\poly(\abs{\Phi(x)})}}$.
\end{proposition}

As a first step, we consider projections of sets of solutions of
systems of linear Diophantine inequalities. To this end, let
$S:A\vec{x}\ge \vec{c}$ be such a system. From
Proposition~\ref{lem:pottier-bound}, we have that
$\eval{S}=\bigcup_{i\in I} L(\vec{b}_i;P_i)$ for some index set
$I$. Let $M_i$ be the projection of $L(\vec{b}_i;P_i)$ on the first
component. We get that $M_i$ can be obtained as
\begin{align*}
  M_i & = \{ b_i + \lambda_{i,1} p_1 + \cdots \lambda_{r}p_{i,r_i} : 
  \lambda_k\in \N\}\\
  & = b_i + g_i \cdot \{ \lambda_1 p_{i,1}/g_i + \cdots
  \lambda_{r} p_{i,r_i}/g_i : \lambda_k \in \N\}
\end{align*}
for some $b_i$, $p_{i,1}<\cdots < p_{i,r_i}$ and $g_i = \gcd\{
p_{i,1},\ldots p_{i,r_i}\}$. Since $\gcd\{ p_{i,1}/g_i, \ldots
p_{i,r_i}/g_i \}=1$, it is folklore that
\begin{align*}
  M_i = b_i + g_i \cdot U(t_i'+1,1,B_i',\{0\}) 
\end{align*}
for some $B_i'\subseteq [t_i']$ and $t_i'\in \N$ known as the Frobenius
number of $p_{i,1}/g_i,\ldots p_{i,r_i}/g_i$. Given co-prime positive
integers $1<a_1<a_2<\cdots < a_k\in \N$, the \emph{Frobenius number
  $f\in \N$} is the largest positive integer not expressible as a
positive linear combination of $a_1,\ldots a_k$ and can be bounded as
follows.
\begin{proposition}[Wilf~\cite{Wilf78}]
  \label{prop:frobenius-bounds}
  Let $1<a_1<a_2<\cdots<a_k\in \N$ be pairwise co-prime. Then the
  Frobenius number $f$ is bounded by $f\le a_k^2$.
\end{proposition}
Hence, for some $t_i \le b_i + g_i(p_{i,r_i}/g_i)^2 \le b_i +
p_{i,r_i}^2$ we consequently have
\begin{align}
  \label{eqn:projection-ups}
  M_i = U(t_i,p_i,B_i,\{0\})
\end{align}
for some $p_i\le p_{i,r_i}$.

Let $\Phi(x)=\exists \vec{x}.\varphi(x,\vec{x})$ be a PA($1$)
formula. From Algorithm~\ref{alg:satisfiability} we can derive that
\begin{align*}
  \eval{\varphi(x,\vec{x})} & = \bigcup_{j\in J} \eval{S_j} = \bigcup_{i\in I}
  L(\vec{b}_i;P_i),
\end{align*}
where each $S_j: A_j(x,\vec{x}) \ge \vec{c}_j$ is a system of linear
Diophantine inequalities obtained from one disjunct of the disjunctive
normal form of $\varphi$, similar as in Line~9 of
Algorithm~\ref{alg:satisfiability}. Clearly,
$\norm{A_j},\norm{\vec{c}_j}\le \norm{\Phi(x)}+1$ for all $i\in
I$. Moreover, from Proposition~\ref{lem:pottier-bound} we derive that
$\eval{S_j}=\bigcup_{i\in I_j} L(\vec{b}_i;P_i)$ for some index set
$I_j$ such that for every $i\in I_j$
\begin{align*}
  \norm{\vec{b}_i}, \norm{P_i} \le (\abs{\Phi(x)}\norm{A_i} +
  \norm{\vec{c}_i} + 1)^{\bigo(\abs{\Phi(x)})} \in
  2^{\poly(\abs{\Phi(x)})}.
\end{align*}

Let $M_i$ be as above, from~(\ref{eqn:projection-ups}) we have
$M_i=U(t_i,p_i,B_i,\{0\})$. Now define $p \defeq \lcm\{ p_i \}_{i\in
  I}$, combining the estimations in~(\ref{eqn:projection-ups}) with
Proposition~\ref{prop:nair} we have $p \in
2^{2^{\poly(\abs{\Phi(x)})}}$. It follows that $\eval{\Phi(x)} =
U(t,p,B,R)$ for some $t\in 2^{\poly(\abs{\Phi(x)})}$ and $p\in
2^{2^{\poly(\abs{\Phi(x)})}}$ as above, which concludes the proof of
Theorem~\ref{thm:qfpa-projections}.

\section{Conclusion}
\label{sec:conclusion}

In the first part of this paper we have shown that Presburger
arithmetic with a fixed number of $i+1$ quantifier alternations and an
arbitrary number of variables in each quantifier block is complete for
$\phsigma_{i}^{\EXP}$ for every $i>0$. This result closes a gap that
has been left open in the literature, and in particular improves and
generalises results obtained by F\"urer~\cite{F82},
Gr\"adel~\cite{Grae89} and Reddy \& Loveland~\cite{RL78}. Moreover, it
provides an interesting natural problem which is complete for the weak
EXP hierarchy, a complexity class for which not that many natural
complete problems have been known so far.

In the second part, we established bounds on ultimately periodic sets
definable in the $\Sigma_1$-fragment of Presburger arithmetic and
showed that in particular the period of those sets is at most
doubly-exponential and that this bound is tight. As already discussed
in the introduction, there are however natural ultimately periodic
sets definable in this fragment that admit periods that are at most
singly-exponential, \emph{cf.}~\cite{GHOW12}. An interesting open
question is whether it is possible to identify a fragment of
$\Sigma_1$-Presburger arithmetic for which such a singly-exponential
upper bound can be established in general and that captures sets such
as those considered in~\cite{GHOW12}.

\section*{Acknowledgments} 
The author would like to thank the anonymous referees for their
thoughtful comments on the first version of this paper. In addition,
the author is grateful to Benedikt Bollig, Stefan G\"oller, Felix
Klaedtke, Sylvain Schmitz and Helmut Veith for encouraging discussions
and helpful suggestions.

\bibliographystyle{plain} \bibliography{bibliography}

\newpage
\begin{appendix}
  \section{Missing proofs}
  In the following, let $\mathbb{B}=\{0,1\}$. Let us recall the
  following characterisation of the polynomial-time hierarchy.
  \begin{lemma}[Def.\ 5.3 and Thm.\ 5.12 in \cite{AB09}]
    \label{lem:p-hierarchy-characterisation}
    For $i>0$, a language $L\subseteq \mathbb{B}^*$ is in
    $\phsigma^\P_i$ iff there exists a polynomial $r$ and a
    deterministic polynomial-time computable predicate $S\subseteq
    (\mathbb{B}^*)^{i+1}$ such that $w \in L$ iff
    \begin{footnotesize}
      \begin{align*}
        \exists w_1\in \mathbb{B}^{r(n)}.\forall
        w_2\in \mathbb{B}^{r(n)}\cdots Q_i w_i\in
        \mathbb{B}^{r(n)}.S(w,w_1,\ldots, w_i).
      \end{align*}
    \end{footnotesize}
  \end{lemma}

  \begin{lemma}[Lem.~\ref{lem:wexp-hierarchy-characterisation} in the main text]
    For any $i>0$, a language $L\subseteq \{0,1\}^*$ is in
    $\phsigma^{\EXP}_i$ iff there exists a polynomial $q$ and a
    predicate $R\subseteq (\{0,1\}^*)^{i+1}$ such that for any $w\in
    \{0,1\}^n$,
    \begin{footnotesize}
    \begin{multline*}
      w\in L \text{ iff } \exists w_1\in \{0,1\}^{2^{q(n)}}.\forall w_2\in
      \{0,1\}^{2^{q(n)}}\cdots\\\cdots Q_i w_i\in \{0,1\}^{2^{q(n)}}.R(w,w_1,\ldots, w_i)
    \end{multline*}
    \end{footnotesize}and $R(w,w_1,\ldots,w_i)$ can be decided in 
    deterministic polynomial time.
  \end{lemma}
  \begin{proof}
    (``$\Leftarrow$'') We describe a $\NEXP^{\phsigma_{i-1}^\P}$
    Turing machine $M$ deciding for a given $w\in \mathbb{B}^n$
    whether $w\in L$. First, $M$ performs a $\NEXP$ guess in order to
    guess $w_1\in \mathbb{B}^{2^{q(n)}}$. Define $L'\subseteq
    \mathbb{B}^n\times \mathbb{B}^{2^{q(n)}}$ such that $(w,w_1)\in L'$ iff
    \begin{footnotesize}
      \begin{align*}
        \exists w_2\in \mathbb{B}^{2^{q(n)}} \cdots Q_i' w_i\in
        \mathbb{B}^{2^{q(n)}}. \neg R(w,w_1,\ldots,w_i),
      \end{align*}
    \end{footnotesize}where $Q_i'=\exists$ if $Q_i=\forall$ and \emph{vice versa}. By
    Lemma~\ref{lem:p-hierarchy-characterisation}, we have that $L'$ is
    a language in $\phsigma_i^\P$ since we can check if the input is
    sufficiently large and immediately reject if this is not the case,
    choose $r: w \mapsto \abs{w} - n$, and decide $\neg
    R(w,w_1,\ldots,w_i)$ in deterministic polynomial time. Thus, after
    $M$ has guessed $w_1$, it invokes the $\phsigma_{i-1}^\P$ oracle
    to check $(w,w_1)\in L'$ and accepts if $(w,w_1)\not\in L'$.

    (``$\Rightarrow$'') Let $L$ be decided by a
    $\NEXP^{\phsigma_{i-1}^\P}$ Turing machine $M$. Given $w\in
    \mathbb{B}^n$, an accepting run of $M$ has length at most
    $2^{n^k}$ for some $k>0$ on which it resolves $c_1,\ldots,c_m \in
    \mathbb{B}, m\le 2^{n^k}$ non-deterministic choices. Moreover, $M$
    makes $\ell$ oracle queries ``$v_j\in L'$?'' for some $L'$ in
    $\phsigma_{i-1}^\P$ such that $n_j=\abs{v_j}, \ell\le 2^{n^k}$,
    and $M$ receives answers $a_j\in \mathbb{B}$ to those queries. By
    Lemma~\ref{lem:p-hierarchy-characterisation}, we have $v_j \in L'$
    iff
    \begin{footnotesize}
      \begin{align*}
        \exists w_{2,j}\in \{0,1\}^{r(n_j)} \cdots
        Q_{i} w_{i,j}\in \{0,1\}^{r(n_j)}. S(v_j,w_{2,j},\ldots,w_{i,j}).
      \end{align*}
    \end{footnotesize}If $M$ receives $a_j=1$ as an answer to an oracle call it can
    guess the corresponding certificate $w_{2,j}\in
    \mathbb{B}^{r(n_j)}$. Otherwise, if $a_j=0$ this result
    can be verified using one quantifier alternation. Consequently, we
    can guess the answers to the oracle queries and then verify at
    once whether those guesses were correct. Hence, $w\in L$ iff
    \begin{small}
      \begin{multline*}
        \exists c_1,\ldots c_m\in \mathbb{B}, 
        v_1,\ldots v_\ell \in \mathbb{B}^{2^{n^k}},
        a_1,\ldots a_\ell\in \mathbb{B},\\
        w_{2,1},\ldots w_{2,\ell}\in \mathbb{B}^{r(2^{n^k})}.
        \forall w_2\in \mathbb{B}^{2^{q(n)}} 
        \cdots Q_{i} w_i\in \mathbb{B}^{2^{q(n)}}.\\
        R((w\cdot c_1\cdots c_m\cdot v_1 \cdots v_\ell\cdot a_1\cdots a_\ell\cdot w_{2,1}\cdots 
        w_{2,\ell}), w_2,\ldots, w_i)
      \end{multline*}
    \end{small}for some appropriately chosen polynomial $q$ and appropriately
    constructed $R\subseteq (\mathbb{B}^*)^{i+1}$ combining $S$ with
    checking that the $c_i$ resolve the non-determinism of $M$
    correctly and that the guessed answers to the oracle calls are
    correct.
  \end{proof}
\end{appendix}

\end{document}